\documentclass[copyright,creativecommons]{eptcs}
\providecommand{\event}{Gandalf 2023} 

\usepackage{amsthm}
\usepackage{amsmath}
\usepackage{amssymb}
\usepackage{iftex}

\ifpdf
  \usepackage{underscore}         
  \usepackage[T1]{fontenc}        
\else
  \usepackage{breakurl}           
\fi

\title{A Uniform One-Dimensional Fragment with Alternation of Quantifiers}
\author{Emanuel Kiero\'nski
\institute{Institute of Computer Science\\University of Wroc\l{}aw, Poland}
\email{emanuel.kieronski@cs.uni.wroc.pl}}

\def\titlerunning{A Uniform One-Dimensional Fragment with Alternation of Quantifiers}
\def\authorrunning{E. Kiero\'nski}

\begin{document}

\newtheorem{theorem}{Theorem}
\newtheorem{lemma}[theorem]{Lemma}
\newtheorem{fact}[theorem]{Fact}
\newtheorem{corollary}[theorem]{Corollary}
\newtheorem{proposition}[theorem]{Proposition}
\newtheorem{claim}[theorem]{Claim}


\newcommand{\FO}{\mbox{\rm FO}}
\newcommand{\FOt}{\mbox{$\mbox{\rm FO}^2$}}
\newcommand{\FOthree}{\mbox{$\mbox{\rm FO}^3$}}
\newcommand{\TGF}{\mbox{$\mbox{\rm TGF}$}}
\newcommand{\GFcp}{\mbox{$\mbox{\rm GF}^{\times_2}$}}
\newcommand{\GFU}{\mbox{$\mbox{\rm GFU}$}}
\newcommand{\UNFO}{\mbox{\rm UNFO}}
\newcommand{\LGF}{\mbox{\rm LGF}}
\newcommand{\GNFO}{\mbox{\rm GNFO}}
\newcommand{\GFTG}{\mbox{\rm GF+TG}}
\newcommand{\GFEG}{\mbox{\rm GF+EG}}
\newcommand{\TGFTG}{\mbox{\rm TGF+TG}}
\newcommand{\GFUTG}{\mbox{\rm GFU+TG}}
\newcommand{\GF}{\mbox{\rm GF}}
\newcommand{\GFt}{\mbox{$\mbox{\rm GF}^2$}}
\newcommand{\FF}{\mbox{\rm FF}}
\newcommand{\UF}{\mbox{$\mbox{\rm UF}_1$}}
\newcommand{\UFAm}{\mbox{$\mbox{\rm AUF}_1^{-}$}}
\newcommand{\UFA}{\mbox{$\mbox{\rm AUF}_1$}}
\newcommand{\UFAthree}{\mbox{$\mbox{\rm AUF}_1^3$}}
\newcommand{\sUF}{\mbox{$\mbox{\rm sUF}_1$}}
\newcommand{\bK}{\mbox{$\bar{\mbox{\rm K}}$}}

\newcommand{\NLogSpace}{\textsc{NLogSpace}}
\newcommand{\LogSpace}{\textsc{LogSpace}}
\newcommand{\NP}{\textsc{NPTime}}
\newcommand{\PTime}{\textsc{PTime}}
\newcommand{\PSpace}{\textsc{PSpace}}
\newcommand{\ExpTime}{\textsc{ExpTime}}
\newcommand{\ExpSpace}{\textsc{ExpSpace}}
\newcommand{\NExpTime}{\textsc{NExpTime}}
\newcommand{\TwoExpTime}{2\textsc{-ExpTime}}
\newcommand{\TwoNExpTime}{2\textsc{-NExpTime}}
\newcommand{\APSpace}{\textsc{APSpace}}
\newcommand{\AExpSpace}{\textsc{AExpSpace}}
\newcommand{\ASpace}{\textsc{ASpace}}
\newcommand{\DTime}{\textsc{DTime}}

\newcommand{\set}[1]{\{#1\}}
\newcommand{\str}[1]{{\mathfrak{#1}}}
\newcommand{\restr}{\!\!\restriction\!\!}
\newcommand{\N}{{\mathbb N}}
\newcommand{\Z}{{\mathbb Z}} 
\newcommand{\sss}{\scriptscriptstyle}
\newcommand{\cT}{\mathcal{T}}
\newcommand{\cL}{\mathcal{L}}
\newcommand{\cF}{\mathcal{F}}
\newcommand{\tpstar}{{\rm tp}^*}
\newcommand{\tp}{{\rm tp}}

\newcommand*{\UU}{\mathsf{U}}
\newcommand*{\NN}{\mathbb{N}}
\newcommand*{\ZZ}{\mathbb{Z}}
\newcommand{\TG}{\mbox{\rm TG}}
\newcommand*{\Lc}{\mathcal{L}}
\newcommand*{\Pc}{\mathcal{P}}
\newcommand*{\Fc}{\mathcal{F}}
\newcommand*{\Mf}{\mathfrak{M}}
\newcommand*{\Nf}{\mathfrak{N}}
\newcommand*{\qf}{\mathrm{qf}}
\newcommand*{\<}{\langle}
\renewcommand*{\>}{\rangle}
\renewcommand{\le}{\leqslant}
\renewcommand{\ge}{\geqslant}
\renewcommand{\bar}{\overline}

\newcommand{\noproof}{\pushQED{\qed} \qedhere \popQED}

\newcommand{\Qfr}{\mbox{Q}}

\newcommand{\phie}{\phi^{\sss \exists}}
\newcommand{\phiu}{\phi^{\sss \forall}}

\newcommand{\psie}{\psi^{\sss \exists}}
\newcommand{\psiu}{\psi^{\sss \forall}}
\newcommand{\mse}{m_{\sss \exists}}
\newcommand{\msu}{m_{\sss \forall}}

\renewcommand{\phi}{\varphi}

\newcommand{\type}[2]{{\rm tp}^{{#1}}({#2})}

\newcommand{\AAA}{\mbox{\large \boldmath $\alpha$}}
\newcommand{\BBB}{\mbox{\large \boldmath $\beta$}}

\maketitle

\begin{abstract}
The uniform one-dimensional fragment of first-order logic was introduced a few years ago as a generalization
of the two-variable fragment of first-order logic to contexts involving relations of arity greater than two. 
Quantifiers in this logic are used in blocks, each block consisting only of existential quantifiers or only of universal 
quantifiers. In this paper we consider the possibility of mixing quantifiers in blocks. We identify
a non-trivial variation of the logic with mixed blocks of quantifiers which retains some  good properties of the two-variable 
fragment and of the uniform one-dimensional fragment: it has the finite (exponential) model property and hence decidable, \NExpTime-complete satisfiability problem.
\end{abstract}

\section{Introduction}
In this paper we are going to push forward the research on the uniform one-dimensional fragment of first-order logic. To
set up the scene and locate our results in a broader context let us first recall some  facts about the two-variable fragment, \FOt{}. \FOt{}, obtained just by restricting first-order logic so
that its formulas may use only variables $x$ and $y$,  is one of the most important decidable fragments of first-order logic identified so far.
The decidability of its satisfiability problem was shown by Scott \cite{Sco62} in the case without equality, and
by Mortimer \cite{Mor75} in the case with equality. In \cite{Mor75} it is proved that the logic has the finite model property, that is, its
every satisfiable formula has a finite model. Later, Gr\"adel, Kolaitis and Vardi \cite{GKV97} strengthened that result, by showing that every satisfiable formula has a model of size bounded exponentially in its length. This exponential model property led to  the \NExpTime{} upper bound on the complexity of \FOt{} satisfiability. The matching lower bound follows from the earlier work by Lewis \cite{Lew80}.

An important  motivation for studying \FOt{} is the fact that it embeds, via the so-called standard translation, basic modal logic and many standard description logics. 
Thus \FOt{} constitutes an elegant first-order framework for those formalisms. However, its simplicity and naturalness make it also
an attractive logic in itself, inheriting potential applications in knowledge representation, artificial intelligence, or verification of
hardware and software from modal and description logics. Plenty of results on \FOt{}, its extensions and variations
have been obtained in the last few decades, e.g., decidability was shown for \FOt{} with counting quantifiers \cite{GOR97,PST97,P-H10}, one or two equivalence relations \cite{KO12,KMPHT14}, counting quantifiers and equivalence relation \cite{PH15}, betweenness relations \cite{KLPS20}, its complexity was established on words and trees, in various scenarios including the presence of data or counting \cite{BMS09,BDM11,CW13,CW15,BBC16},  to mention just a few of them.

 However, further applications, e.g., in database theory are limited by the fact that \FOt{} and its extensions mentioned above
 can speak non-trivially only about
relations of arity at most two. This is in contrast to some other decidable fragments studied because of their potential applications in computer science, 
like the guarded fragment, \GF{} \cite{ABN98},
the unary negation fragment, \UNFO{} \cite{StC13}, the guarded negation fragment, \GNFO{} \cite{BtCS15}, or the fluted fragment, \FF{} \cite{Q69,P-HST19}. 

A natural question is whether there is an elegant decidable formalism which retains full expressivity of \FOt{}, but additionally, allows one to speak
non-trivially about relations of arity  bigger than two. In the recent literature we 
can find a few such formalisms. 

An interesting idea is for example to combine  \FOt{} with \GF{}. The idea can be traced back already in Kazakov's PhD thesis \cite{Kaz06}, was present in the  work
by Bourish, Morak and Pieris  \cite{BMP17}, and found a more systematic treatment in the paper by 
Rudolph and \v{S}imkus \cite{RS18}, who  formally introduced the triguarded fragment, \TGF{}.  \TGF{} is obtained from
\GF{} by allowing quantification for subformulas with at most two free variables to be unguarded.
What we get this way is a logic in which one can  speak freely about pairs of elements, and in a local, guarded way about
tuples of bigger arity. 
\TGF{} turns out to be undecidable with equality, but becomes decidable when equality is forbidden. The satisfiability problem
is then \TwoExpTime{}- or \TwoNExpTime-complete, depending on whether constants are allowed in signatures \cite{RS18}; the finite model property is
retained \cite{KR21}.
A variation of the idea above is the one-dimensional triguarded fragment \cite{Kie19}, still containing \FOt{}, which becomes decidable even in the presence of equality.

\FOt{} (or, actually, even its extension with counting quantifiers, C$^2$) was also combined with   \FF{} by
 Pratt-Hartmann \cite{P-H21}. This logic was shown decidable but the complexity of its satisfiability problem is non-elementary,
as already \FF{} alone has non-elementary complexity
\cite{P-HST19}.

Finally, probably the most canonical extension of \FOt{} to contexts with relations of arity bigger than two, capturing the spirit of \FOt{} more closely than
the logics discussed above, is the \emph{uniform one-dimensional fragment}, \UF{}, proposed by Hella and Kuusisto \cite{HK14}.
 In this fragment quantifiers are used in blocks and a single block is built out only of existential or only of universal quantifiers and leaves
at most one variable free; a fragment meeting this condition is called \emph{one-dimensional}. Imposing one-dimensionality alone is not sufficient for ensuring the
decidability of the satisfiability problem and thus another restriction, \emph{uniformity}, is applied which, roughly speaking,
allows boolean combinations of atoms only if the atoms use precisely the same set of variables or use just one variable.
In effect, just as \FOt{} contains modal logic (or even Boolean modal logic), \UF{} contains \emph{polyadic} modal logic (even with negations of the accessibility relations) (cf.~\cite{Kuu16}).
 In \cite{HK14} it is shown that
\UF{} without equality is decidable and has the finite model property. In \cite{KK14} this result is improved by showing that the decidability
is retained even if  {free} use of equalities is
allowed (by \emph{free use} of equalities we mean that they need not obey the uniformity restriction) and that the logic has exponential model property and  \NExpTime-complete satisfiability problem, exactly as \FOt{}.

A question arises whether the requirement that the blocks of quantifiers from the definition of \UF{} must consist of quantifiers of the same type (all universal or all existential) is necessary
for decidability, that is what happens if we allow one to mix quantifiers as, e.g., in the formula $\forall x \exists y \forall z R(x,y,z,t)$.
Let us denote the
extension of \UF{} allowing to alternate  quantifiers in blocks  $\UFA$. 
The motivations behind studying \UFA{} are multifarious. \UF{} lies very close to the borderlines between the decidable and undecidable, so,
firstly and most importantly, analysing its expressive extensions may enhance our understanding of these borderlines which may be also useful in different scenarios.
Secondly, the logics \UF{} and \UFA{} can be useful themselves, offering  extensions of modal and description logics
to contexts with relations of arity greater than two, such as databases, orthogonal to other proposals. Thirdly, though it is of course a matter of 
taste, we believe that \UFA{} is just quite an elegant formalism, which can be justified by a relative simplicity of its definition and a nice game-theoretic characterization of its expressivity---natural Ehrenfeucht-style games for \UF{} were introduced in \cite{KK14}; shifting to \UFA{} would probably allow for an even nicer game characterizations
(though this topic is not formally studied in this paper).

The first step to understand \UFA{} was done in the companion paper \cite{FK23},  where we show the decidability and the finite model property of the three variable restriction of this logic, \UFAthree{}; in that paper \UFAthree{} is then  made a basis for obtaining a rich decidable subclass of the three-variable fragment, \FOthree{}.

Turning now to our current contribution, we first remark that in this paper we still do not answer the question whether the whole \UFA{} has decidable satisfiability. We however make another step towards understanding \UFA{} by identifying its fragment,  \UFAm{}, which contains full \FOt{} without equality, allows for mixed blocks of quantifiers of unbounded length,
has \NExpTime-complete satisfiability
problem, and has the exponential model property. Additionally, we observe that if we allow for a free use of equality in \UFAm{} then we lose the 
finite model property. 

The main restriction of \UFAm{}, compared to full \UFA{}, is that it admits only blocks of quantifiers that are purely universal or end with the existential quantifier. Additionally, mostly for the clarity of presentation, we will define \UFAm{} not as an extension of the version of \UF{} originally defined in \cite{HK14}, 
but rather as an extension of  the \emph{strongly} uniform one-dimensional fragment, \sUF{}, introduced in \cite{KK15}. 
The definition of \UFAm{} is inspired by the definition of the Maslov class $\bK$ \cite{Mas71} and, as we will see in a moment, the
decidability of \UFAm{} can be shown by a reduction to conjunctions of sentences in $\bK$, whose decidability was
shown by the resolution method by Hustadt and Schmidt \cite{HS99}.
However, this reduction  does not allow us to establish the precise compleixty of \UFAm{}, since, to the best of our knowledge, the precise complexity of the Maslov class has not been established. It is also not known whether \bK{} has the finite model property.

\section{Preliminaries}

\subsection{Notation and terminology} 
We assume that the reader is familiar with first-order logic.
We work with purely relational signatures with no constants nor  function symbols. 
We refer to structures using Fraktur capital letters, and to their domains using
the corresponding Roman capitals. Given a structure $\str{A}$ and some $B \subseteq A$ we
denote by $\str{A} \restr B$ the restriction of $\str{A}$ to its subdomain $B$. 

We usually use $a, b, \ldots$ to denote elements of structures, and $x$, $y$, $\ldots$ for
variables; all of these possibly with some decorations.
For a tuple of variables $\bar{x}$ we use $\psi(\bar{x})$ to denote that the free variables
of $\psi$ are in $\bar{x}$.

In the context of  uniform logics it is convenient to speak about some partially defined (sub)structures which we will call \emph{pre-(sub)structures}.
A pre-structure over a signature $\sigma$ consists of its domain $A$ and a function specifying the truth-value of every fact
$P(\bar{a})$, for $P \in  \sigma$ and a tuple $\bar{a}$ of elements of $A$ of length equal to the arity of $P$, such that 
$\bar{a}$ contains all elements of $A$ or just one of them. The truth values of all the other facts remain unspecified.
We will use Fraktur letters decorated with $*$ to denote pre-structures:
a pre-structure with domain $A$ will be denoted by $\str{A}^*$. If a structure $\str{A}$ is fully defined, $\str{A}^*$ denotes its induced pre-structure.
Similarly, if $B \subseteq A$ is a subdomain of some structure $\str{A}$ we donote by $\str{B}^*$ the pre-structure $(\str{A} \restr B)^*$
and call it a pre-substructure of $\str{A}$.

An (atomic) $1$-{\em type} over a signature $\sigma$ is a
maximal consistent set of atomic or negated atomic formulas  over
$\sigma$ using at most one variable $x$. We often identify a $1$-type with the conjunction of its elements. 
We will usually be interested in $1$-types over signatures $\sigma$ consisting of the
relation symbols used in some given formula. Observe that the number of $1$-types is bounded by a function which is exponential in $|\sigma|$, and hence also in the length
of the formula. This is because a $1$-type just corresponds to a subset of $\sigma$.

Let $\str{A}$ be a structure, and let
$a \in A$. We denote by $\type{\str{A}}{a}$ the unique atomic
1-type \emph{realized} in $\str{A}$ by the element $a$, \emph{i.e.}, the $1$-type $\alpha(x)$ such that $\str{A} \models \alpha(a)$.

\subsection{Satisfiability and finite model property}

Let $\mathcal{L}$ be a class of first-order formulas (a logic). The \emph{(finite) satisfiability problem} for $\mathcal{L}$
takes as its input a sentence from $\mathcal{L}$ and verifies if it has a (finite) model.
$\mathcal{L}$ has the \emph{finite model property} if  every satisfiable sentence in $\mathcal{L}$ has a finite model; $\mathcal{L}$ has the \emph{exponential model property} if there is a fixed exponential function $f$ such that
every satisfiable sentence $\varphi$ has a finite model over a domain whose size if bounded by $f(|\varphi|)$
(where the length of $\varphi$, $|\varphi|$, is measured in any reasonable fashion).

\subsection{Logics}

As the starting point we define the logic s\UF{} (without equality), called in \cite{KK15} the \emph{strongly restricted uniform one-dimensional fragment}.
Formally, for a relational signature $\sigma$, the set of $\sigma$-formulas of s\UF{} is the smallest set  $\cF$ such that:
\begin{itemize} \itemsep0pt
\item every $\sigma$-atom using at most one variable is in $\cF$
\item $\cF$ is closed under Boolean connectives
\item if $\phi(x_0, \ldots, x_k)$ is a Boolean combination of formulas in $\cF$ with free variables in $\{x_0, \ldots, x_k\}$
and atoms\footnote{Please note that those atoms need not to belong to $\mathcal{F}$.} built out of precisely all of the variables $x_0, \ldots, x_k$ (in an arbitrary order, possibly with repetitions)
then $\exists x_0, \ldots, x_k \phi$, $\exists x_1, \ldots, x_k \phi$,  $\forall x_0, \ldots, x_k \phi$ and $\forall x_1, \ldots, x_k \phi$ are in $\cF$.  
\end{itemize}

Example formulas in s\UF{} are: 
$$\forall xyz (P(x) \wedge P(y) \wedge P(z) \rightarrow R(x,y,z) \vee \neg S(z,z,x,y))$$
$$\forall x (P(x) \rightarrow \exists yz (\neg R(y,z,x) \wedge (\neg R(x,y,z) \vee P(y))))$$

For interested readers we say that (non-restricted) \emph{uniform one-dimensional fragment} is defined as above but in the last point of the definition
the non-unary atoms must not  necessarily use the whole set $\{x_0, \ldots, x_k\}$ of variables but rather all those atoms use the same subset of this set (see \cite{HK14}).

By \UFA{} we denote the extension of s\UF{} without equality with \emph{alternation of quantifiers in blocks}. The set of $\sigma$-formulas of \UFA{} is the smallest
set $\cF$ such that: 

\begin{itemize} \itemsep0pt
\item every $\sigma$-atom using at most one variable is in $\cF$
\item $\cF$ is closed under Boolean connectives
\item if $\phi(x_0, \ldots, x_k)$ 
is a Boolean combination of formulas in $\cF$ with free variables in $\{x_0, \ldots, x_k\}$
and atoms built out of precisely all of the variables $x_0, \ldots, x_k$ (in an arbitrary order, possibly with repetitions)
then $\Qfr_0 x_0  \ldots \Qfr_k x_k \phi$ and $\Qfr_1 x_1 \ldots \Qfr_k x_k \phi$ are in $\cF$, where each $\Qfr_i$ is one of $\exists, \forall$.  
\end{itemize}

Finally, we define a subset \UFAm{} of \UFA{} by requiring that its formulas are written in negation normal form NNF (that is negation is used only
in front of atomic formulas, and the only other Boolean connectives are $\vee$ and $\wedge$), and that every sequence of quantifiers in the last point of the definition
either contains only universal quantifiers or the last quantifier $\Qfr_k$ is existential. In \UFAm{} we can write, e.g.:
$$\forall xy \exists z (\neg P(x) \wedge \neg P(y) \vee R(x,y,z))$$
$$\forall x (P(x) \vee \exists y \forall z \exists t S(x,y,z,t))$$
$$\forall xyz ( R(x,y,z) \vee \exists t T(x,t) \wedge \exists t T(y,t) \wedge \exists t T(z,t))$$ 

Observe that \UFAm{} contains the whole s\UF{} and hence also \FOt{}. (For example the \FOt{} sentence $\exists x \forall y \psi(x,y)$ belongs
to s\UF{}, since one may think that it has two blocks of quantifiers, both of length one, and indeed one of them is purely universal and
the other ends with $\exists$.)

\subsection{Normal forms and basic decidability result}
We introduce normal form for \UFAm{} formulas, generalizing Scott's normal form for \FOt{} (cf.~\cite{Sco62,GKV97}). We start with a version
involving $0$-ary predicates, called \emph{weak} normal form, and then explain how to remove them. In our normal form as well as
in some intermediate formulas 
we allow ourselves to use implications which are usually not allowed in NNF formulas, but here they are very natural (note 
that converting them to disjunctions using the basic law $p \rightarrow q \equiv \neg p \vee q$ will not affect the blocks
of quantifiers).
We say that a \UFAm{} sentence is in \emph{weak normal form} if it is a conjunction of formulas having one of the following shapes.
\begin{eqnarray} 
&&\Qfr_1x_1 \ldots \Qfr_kx_k \psi(x_1,\ldots, x_k),
\label{eq:nf1a}\\
&&E \rightarrow  \Qfr_1x_1 \ldots \Qfr_kx_k \psi(x_1,\ldots, x_k),
\label{eq:nf1b}
\end{eqnarray}
where
\begin{itemize} \itemsep0pt
 \item $k \ge 0$ is a natural number 
  \item the $x_i$ are distinct variables
	\item $\psi$ is a quantifier-free \UFAm{} formula (Boolean combination of atoms, each of them containing all the variables
	$x_1, \ldots, x_k$ or at most one of them)
	\item  every $\Qfr_i$ is a quantifier (universal or existential) and either all the $\Qfr_i$ are universal (\emph{universal conjunct}) or $\Qfr_k$ is existential (\emph{existential conjunct})
	\item $E$ is a $0$-ary relation symbol
\end{itemize}

In particular, in a formula of type (\ref{eq:nf1a}), $k$ may be equal to $0$; in this case $\psi$ is a Boolean combination of $0$-ary predicates.
\begin{lemma} \label{l:nf}
Let $\phi$ be a \UFAm{} sentence. Then there exists a polynomially computable  \UFAm{} sentence $\phi'$ in weak normal form over a signature extending the
signature of $\phi$ by  some fresh unary and $0$-ary relation symbols, such that (i) every model of $\phi$ can be expanded to a model of $\phi'$ and (ii) every 
model of $\phi'$ is a model of $\phi$. 
\end{lemma}

\begin{proof} (Sketch)
Assume that $\phi$ is in NNF.
Take an innermost subformula $\psi_0$ starting with a maximal block of quantifiers. If it has a free variable, that is, is of the form $$\Qfr_1x_1 \ldots \Qfr_kx_k \psi(x_1,\ldots, x_k,y)$$ 
replace it by $P(y)$, for a fresh unary symbol $P$, and add the following normal form conjunct $\phi_{\psi_0}$ (partially) axiomatising $P$.
$$\forall y \Qfr_1x_1 \ldots \Qfr_kx_k (P(y) \rightarrow \psi(x_1, \ldots, x_k, y)).$$ 
In other words, $\phi$ is replaced by $\phi(P(y)/\psi_0) \wedge \phi_{\psi_0}$. 

If $\psi_0$  is a proper subsentence, that is it is of the form 
$$\Qfr_1x_1 \ldots \Qfr_kx_k \psi(x_1,\ldots, x_k),$$ then replace it by $E$, for a fresh $0$-ary symbol $E$ and add the conjunct
$$E \rightarrow \Qfr_1x_1 \ldots \Qfr_kx_k \psi(x_1,\ldots, x_k).$$

Repeat this process as long as possible. Note that we indeed append conjuncts belonging
to \UFAm{}.

The above process is similar to Scott's reduction of \FOt{} formulas to their normal form. Besides the natural modifications needed to deal with
sequences of quantifiers rather than with single quantifiers, the main difference is that in the appended conjuncts
axiomatizing the freshly introduced unary and $0$-ary predicates, we write
 implications in only one direction. 
This is sound as our initial formula is assumed to be in NNF. Indeed, consider a single step of the reduction, assuming that the case with a free variable in $\psi_0$ applies. In this step $\phi'=\phi(P(y)/\psi_0) \wedge \phi_{\psi_0}$ is produced from $\phi$. Assuming that $\str{A} \models \phi$ we obtain a model $\str{A}'$ of
$\phi'$ by making the unary relation $P$ true at all elements $a$ of $A$ such that $\str{A} \models \psi_0[a]$.
This makes the appended conjunct $\phi_{\psi_0}$  true; obviously, also $\phi(P(y)/\psi_0)$ remains true. 
In the opposite direction assume $\str{A}' \models \phi'$. It may happen that the subformula $\psi_0(y)$  of $\phi$
is true in $\str{A}'$ in more points than $P(y)$ is. However, to guarantee that $\phi$ is true it suffices that it is true \emph{at least} at those
points where $P(y)$ is, which is ensured by the appended conjunct $\phi_{\psi_0}$; this is because $\phi$ is assumed to be in NNF and thus
$\psi_0$ appears in $\phi$ in the scope of no negation symbol. We reason similarly for the case when $\psi_0$ is a subsentence. 
\end{proof}

For our purposes, that is showing the finite model property for \UFAm{} and demonstrating that its satisfiability problem is in \NExpTime{}, we can further simplify our formulas,
by eliminating $0$-ary predicates. What we can do is to guess the truth values for all the $0$-ary predicates and replace them by $\top$ or
$\bot$, in accordance with the guess. In particular, the conjuncts of the form $E \rightarrow  \Qfr_1x_1 \ldots \Qfr_kx_k \psi(x_1,\ldots, x_k)$ are eliminated if $E$ is guessed to be $\bot$ and replaced just by $\Qfr_1x_1 \ldots \Qfr_kx_k \psi(x_1,\ldots, x_k)$ if $E$ is guessed
to be $\top$.

It is convenient to split the set of the resulting conjuncts into those whose all quantifiers are universal and those which end with the existential 
quantifier. We say that a sentence $\phi$ is in \emph{normal form} if it is of the following shape:
\begin{eqnarray} \label{eq:normal}
&&\bigwedge_{1\le i \le \mse} \phie_i
\wedge \bigwedge_{1\le i \le \msu}  \phiu_i,
\end{eqnarray}
where   $\phie_i= \Qfr_1^i x_1 \Qfr_2^i x_2 \ldots \Qfr_{k_i{-}1}^i  x_{k_i{-}1} \exists  x_{k_i} \psie_i,$
$\phiu_i=\forall x_1 \ldots x_{l_i} \psiu_i,$
for  $\Qfr_j^i \in \{ \forall, \exists \}$, $\psie_i=\psie_i(x_1, x_2, \ldots, x_{k_i})$ and $\psiu_i=\psiu_i(x_1, \ldots, x_{l_i})$.

\medskip
The discussion above justifies the following.

\begin{lemma}\label{l:nf2}
(i)
The satisfiability problem for \UFAm{} can be reduced in nondeterministic polynomial time to the satisfiability problem
for normal form \UFAm{} sentences.
(ii) If the class of all normal form \UFAm{} sentences has the finite (exponential) model property then also the whole \UFAm{} has the
finite (exponential) model property. 
\end{lemma}

The reduction to normal form described above allows us to easily prove the decidability of the satisfiability problem for \UFAm{}. 
This can be done by using the results on  the Maslov Class $\bar{\text{K}}$ (which is a dual of the Maslov class $\text{K}$). Full definition of $\bar{\text{K}}$ is quite
complicated and can be found, e.g., in \cite{HS99}. For our purposes it is sufficient to say that when converted to prenex
form $\bar{\text{K}}$ formulas look as follows:
\begin{eqnarray}
\label{eq:maslov}
\exists y_1 \ldots \exists y_m \forall x_1 \ldots \forall x_k \Qfr_1 z_1 \ldots \Qfr_l z_l \psi,
\end{eqnarray}
where the $\Qfr_i$ are quantifiers, $\psi$ is a quantifier-free formula without equality  
and every atom of $\psi$ satisfies one of the following conditions: (i) it contains at most one $x_i$- or $z_i$-variable, (ii) it 
contains all the $x_i$-variables and no $z_i$-variables, or (iii) it contains an existentially quantified  variable $z_j$ and no $z_i$-variables
with $i>j$.

Now, one easily observes that every \UFAm{}-normal form conjunct belongs to $\bar{\text{K}}$.  Indeed, every $\phiu_i$-conjunct is of the form (\ref{eq:maslov}) with $m=l=0$ and its every atom  satisfies either condition (i) or (ii); every $\phie_i$-conjunct is of the form (\ref{eq:maslov}) with $m=k=0$ and every atom satisfying (i) or (iii).

Hence any normal form formula belongs to 
 $\bar{\text{DK}}$,  the class of conjunctions of formulas in $\bar{\text{K}}$. The satisfiability problem for 
$\bar{\text{K}}$ was shown to be decidable in \cite{Mas71}. This result was extended to  the  class $\bar{\text{DK}}$ in \cite{HS99}. This gives us the basic decidability result.
 
\begin{theorem}
The satisfiability problem for \UFAm{} is decidable.
\end{theorem}

We recall that the precise complexity of $\bar{\text{DK}}$ has not been established. It is also not known if 
$\bar{\text{DK}}$ has the finite  model property and if its finite satisfiability is decidable.
The same questions for $\bar{\text{K}}$ are also open.

\section{Finite model property}

The following theorem is the main results of this paper. Besides just proving the finite model property for \UFAm{}, it will also allow us to establish the exact complexity of its satisfiability problem.

\begin{theorem} \label{t:fmp}
\UFAm{} has the exponential model property.  
\end{theorem}

The rest of this section is devoted to a proof of the above theorem.
 By Lemma \ref{l:nf2} we may restrict attention to formulas of the form (\ref{eq:normal}).

\subsection{Satisfaction forests} \label{s:satfor}

In this subsection we introduce \emph{satisfaction forests}, which are auxiliary structures (partially) describing some finite models of normal form 
\UFAm{} sentences.  
We first explain how to extract a satisfaction forest from a given finite model $\str{B}$ of a normal form sentence $\varphi$. Then we
formally define satisfaction forests and relate their existence to the existence of finite models of normal form sentences.

\subsubsection{Extracting a satisfaction forest from a model} 
\label{s:extraction}
Let $\varphi$ be a normal form $\UFAm$ sentence and let $\str{B}$ be its finite model. Assume that $\varphi$ is as in (\ref{eq:normal}). The satisfaction forest will be
a collection of labelled trees, one tree for each existential conjunct of $\varphi$, showing how this conjunct is satisfied in $\str{B}$. The labelling function will be 
denoted $\cL$ and will assign elements from $B$ to tree nodes (with the exception of the root, which will be assigned the special empty label).  

Consider a single existential conjunct 
$\phie_i=\Qfr_1^i x_1 \Qfr_2^i x_2 \ldots \Qfr_{k_i{-}1}^i  x_{k_i{-}1} \exists  x_{k_i} \psie_i$.
Its satisfaction tree $\cT_i$ is built in the following process.

Start with a root labelled with the empty label. The root forms level $0$ of the tree. Level $j$, $0 < j \le k_i$ will correspond
to the quantifier $\Qfr^i_j$. Assume level $j-1$ has been constructed, for $0 < j  \le k_i$. For each of its nodes $n$: 
\begin{itemize} \itemsep0pt
\item If 
$\Qfr^i_j=\forall$ then for each element $b \in B$ add a child $n'$ of $n$ to $\cT_i$ and set $\cL(n'):=b$. Nodes added in this step are called \emph{universal nodes}.
\item If 
$\Qfr^i_j=\exists$ then let $n_1, \ldots, n_{j-1}=n$ be the sequence of non-root nodes on the branch of $n$, ordered from the child of the root towards $n$. Choose in $\str{B}$ an element $b$ such that 
$$\str{B} \models \Qfr_{j+1}^i x_{j+1} \ldots \Qfr_{k_i-1}^i x_{k_i-1}\exists x_{k_i} \psie_i (\cL(n_1), \ldots, \cL(n_{j-1}), b, x_{k+1}, \ldots, x_{k_{i}}).$$
It is clear that such an element exists. If $j<k_i$ then we call it an \emph{intermediate witness} for $\phie_i$, and if $j=k_i$ we call it the \emph{final witness} for $\phie_i$. Add a single child $n'$ of $n$ to $\cT_i$ and set $\cL(n')=b$. The added element is called an \emph{existential node}.
\end{itemize}
For a branch $\flat$ of the above-defined tree
we denote by $Set(\flat)$ the set of labels of the non-root elements of $\flat$, by $Set^-(\flat)$ the set of labels of non-root and non-leaf elements
of $\flat$ and by $Seq(\flat)$ the sequence
of the non-root elements of $\flat$, ordered from the child of the root towards the leaf.

We further overload the function $\cL$ by allowing it to define also labels for branches of the tree (by a \emph{branch} we mean here
a sequence of elements $n_1, \ldots, n_{k_i}$ such that $n_1$ is a child of the root, $n_{k_i}$ is a leaf, and each $n_{i+1}$ is a child of $n_i$).
We label each branch $\flat$ with the pre-substructure of $\str{B}$ over $Set(\flat)$.

To declare some properties of satisfaction forests we need the following notions.
A pre-structure $\str{H}^*$  is $\phiu_i$-\emph{compatible},
if for every sequence $a_1, \ldots, a_{l_i}$ of elements of $H$ such that  $\{a_1, \ldots, a_{l_i}\}=H$ 
we have $\str{H}^* \models \psiu_i(a_1, \ldots, a_{l_i})$.
A pre-structure is $\phiu$-\emph{compatible} if it is $\phiu_i$-compatible for every conjunct 
$\phiu_i$.
Further, a set of $1$-types $\{\alpha_1, \ldots, \alpha_k \}$ is $\phiu$-\emph{compatible} if for a set of  distinct
elements $H=\{a_1, \ldots, a_m \}$ and any assignment $f: \{a_1, \ldots, a_m\} \rightarrow \{\alpha_1, \ldots, \alpha_k \}$
one can build a $\phiu$-compatible pre-structure on $H$ in which, for every $i$, the $1$-type of  $a_i$ is $f(a_i)$.

\medskip
We now collect some properties of the tree $\cT_i$ for $\phie_i$ constructed as above.
\begin{enumerate} \itemsep0pt
\item[(T1)] for $1 \le j \le k_i$, and every node $n$ from level $j-1$:
\begin{enumerate} \itemsep0pt
	\item if $\Qfr_j^i=\forall$ then $n$ has precisely $| B |$ children, labelled by distinct elements of $B$ (recall that each of these children is called a universal node)
	\item if $\Qfr_j^i=\exists$ then $n$ has precisely one child (recall that this child is called an existential node)
	\end{enumerate}
	\item[(T2)] for every branch $\flat \in \cT_i$, assuming $Seq(\flat)=(a_1, \ldots, a_{k_i})$, we have $\cL(\flat) \models \psie_i(a_1, \ldots, a_{k_i})$
	\item[(T3)] for every pair of branches $\flat_1, \flat_2 \in \cT_i$, for every $a \in B$ such that $a \in Set(\flat_1)$ and $a \in Set(\flat_2)$ 
 the $1$-types of $a$ in $\cL(\flat_1)$ and in $\cL(\flat_2)$ are identical
	\item[(T4)] for every pair of branches $\flat_1, \flat_2 \in \cT_i$ such that $Set(\flat_1)=Set(\flat_2)$ we have that $\cL(\flat_1)  = \cL(\flat_2)$.
	\item[(T5)] 	for every branch $\flat \in \cT_i$, $\cL(\flat)$ is $\phiu$-compatible
	\end{enumerate}
	
Now we collect some properties of the whole sequence of trees  $\cT_1, \ldots, \cT_{\mse}$ constructed for $\varphi$ and $\str{B}$.

\begin{enumerate} \itemsep0pt
\item[(F1)] for every $i$, $\cT_i$ is a satisfaction tree over $B$ for $\phie_i$
\item[(F2)] for every pair of branches $\flat_1 \in \cT_i, \flat_2 \in \cT_j$, $i \not=j$, for every $a \in B$ 
such that $a \in Set(\flat_1)$ and $a \in Set(\flat_2)$ 
	the $1$-types of $a$ in $\cL(\flat_1)$ and in $\cL(\flat_2)$ are identical
	\item[(F3)] for every pair of branches $\flat_1 \in \cT_i, \flat_2 \in \cT_j$, $i \not=j$ such that $Set(\flat_1)=Set(\flat_2)$ we have that $\cL(\flat_1)  = \cL(\flat_2)$
	\item[(F4)] the set of all $1$-types appearing in the pre-structures defined as labels of  the branches of the trees in $\cF$ is $\phiu$-compatible.

\end{enumerate}

Properties (T3), (T4), (F2) and (F3) will be sometimes called the \emph{(forest) consistency conditions}. 

\begin{claim}\label{c:conds}
The sequence of trees $\cT_1, \ldots, \cT_{\mse}$ constructed as above for the structure $\str{B}$ and the sentence $\varphi$ satisfies conditions (T1)-(T5) and (F1)-(F4).
\end{claim}
\begin{proof} (Sketch)
It is not difficult to see that each of the $\cT_i$ satisfies (T1)-(T5)  and that the whole sequence satisfies (F1)-(F3). 
The only non-obvious point is (F4). Let us prove that it is true. Let $\alpha_1, \ldots, \alpha_k$ be the list of all
$1$-types appearing in the pre-structures defined in the whole forest. Let $H=\{a_1, \ldots, a_m \}$ be a set of fresh distinct elements
and $f: \{a_1, \ldots, a_m \} \rightarrow \{\alpha_1, \ldots, \alpha_k \}$ an assignment of $1$-types to these elements. 
We need to construct a
$\phiu$-compatible pre-structure on $H$ in which, for every $i$, the $1$-type of  $a_i$ is $f(a_i)$. 
For each $i$ choose an element $g(a_i) \in B$ such that $\tp^{\str{B}}(g(a_i)) = f(a_i)$; $g$ need not be injective. Let us 
define the pre-structure $\str{H}^*$ on $H$ by setting the $1$-type of $a_i$ to be $f(a_i)$, and
for every relation symbol $R$, and every sequence $c_1, \ldots, c_l$ of elements of $H$ such that $l$ is the arity of $R$ and
$\{c_1, \ldots, c_l\}=H$, setting the truth-value of the atom $R(c_1, \ldots, c_l)$ to be equal to the truth-value 
of  $R(g(c_1), \ldots, g(c_l))$ in $\str{B}$. 
We claim that so defined $\str{H}^*$ is $\phiu$-compatible. To see this take any conjunct $\phiu_i$
and any sequence of elements $c_1, \ldots, c_{l_i}$ such that $\{c_1, \ldots, c_{l_i}\}=H$ and
assume to the contrary that $\str{H}^* \not\models \psiu_i(c_1, \ldots, c_{l_i})$.
But then $\str{B} \not\models \psiu_i(g(c_1), \ldots, g(c_{l_i}))$, as the truth-values of the atoms appearing in $\psiu_i$
in the two considered structures appropriately coincide by our definition of $\str{H}^*$. Contradiction. 
\end{proof}

\subsubsection{Satisfaction forests and the existence of finite models} Let $\varphi$ be a normal form \UFAm{} sentence (we do not assume that a model of $\varphi$ is known). Formally, a \emph{satisfaction forest} for
$\varphi$ \emph{over a domain} $B$ is a sequence of
	trees $\cT_1, \ldots, \cT_{\mse}$ together with a labelling function $\cL$, assigning elements of $B$ to the nodes of the $\cT_i$ (with the exception of their roots to which the special empty label is assigned) and pre-structures to their branches, such that each of the trees $\cT_i$ satisfies conditions
	(T1)-(T5) and the whole sequence satisfies conditions (F1)-(F4).

\begin{lemma} \label{l:modelexists}
A normal form \UFAm{} sentence $\varphi$ has a finite model over a domain $B$ iff it has a satisfaction forest over $B$.
\end{lemma}
\begin{proof}
Left-to-right implication is justified by the extraction of a satisfaction forest from a given finite model of $\varphi$
described in Section \ref{s:extraction}, and in particular by Claim \ref{c:conds}.

In the opposite
direction assume that a satisfaction forest over a finite domain $B$ for $\varphi$ is given. We construct a model $\str{B}$ of $\varphi$
over the domain $B$. The construction is natural:

\medskip\noindent
\emph{Step 1: $1$-types.} The $1$-type of an element $b \in B$ is defined as the $1$-type of $b$ in the structure
$\cL(\flat)$ for an arbitrarily chosen branch $\flat$, in an arbitrarily chosen tree $\cT_i$, for which $b \in Set(\flat)$.

\medskip\noindent
\emph{Step 2: Witnesses.}
For every tree $\mathcal{T}_i$ and  its every branch $\flat$ of $\mathcal{T}_i$ define the pre-structure on $Set(\flat)$ 
in accordance with $\cL(\flat)$. 

\medskip\noindent
\emph{Step 3: Completion.}
For any set of distinct elements $\{b_1, \ldots, b_k\}$ whose pre-structure is not yet defined, choose any
$\phiu$-compatible pre-structure which retains the already defined $1$-types of the $b_i$.

\medskip
Properties (T3), (F2), (T4) and (F3) guarantee that Step 1 and Step 2 can be performed without conflicts
and the existence of an appropriate pre-structure in Step 3 is guaranteed by (F4).

\medskip
It remains to see that $\str{B} \models \varphi$. 
Consider any existential conjunct of $\phi$, that is a conjunct $\phie_i=\Qfr_1^i x_1 \Qfr_2^i x_2 \ldots \Qfr_{k_i{-}1}^i  x_{k_i{-}1} \exists  x_{k_i} \psie_i$.
The satisfaction tree $\cT_i$ witnesses that $\phie_i$ indeed holds: it describes all possible substitutions for universally quantified
variables, and shows how intermediate and final witnesses for existential quantifiers can be chosen.
Consider now any universal conjunct $\phiu_i=\forall x_1 \ldots x_{l_i} \psiu_i$ and
let $b_1, \ldots, b_{l_i}$ be any sequence of elements of $B$ (possibly with repetitions).
Let $H=\{b_1, \ldots, b_{l_i}\}$. The pre-structure on $H$ has been defined either in Step 2 or in Step 3.
In both cases we know that it is $\phiu$-compatible, in particular it is $\phiu_i$-compatible,
so $\str{H}^* \models \psi_i^\forall(b_1, \ldots, b_{l_i})$.  
\end{proof}

\subsection{From a model to a satisfaction forest over a small domain}

We are ready to present the main construction of this paper in which we show that every satisfiable formula has a satisfaction forest over a small domain. 

Let $\str{A}$ be a (possibly infinite) model of a normal form sentence $\phi$ of the shape as in (\ref{eq:normal}). We show how to construct a satisfaction forest over a domain of size bounded exponentially
in $| \phi |$. By Lemma \ref{l:modelexists} this will guarantee that $\phi$ has a finite model over such a bounded domain.

\subsubsection{Domain}
Let $L$ be the number of $1$-types (over the signature of $\phi$) realized in $\str{A}$, and let these types be enumerated as $\alpha_1, \ldots, \alpha_L$. Let $K=\max \{k_i: 1 \le i \le \mse \}$.  
We define the domain $B$ to be $\{1, \ldots, 2K \} \times \{1, \ldots, \mse \} \times \{1, \ldots, (K-1)^{K-1} \} \times \{1, \ldots, L\}$.   Note that $K$ and $\mse$ are bounded linearly  and $L$ is bounded exponentially in $| \phi | $, and hence $ | B |$ is indeed bounded exponentially in $| \phi |$.

For convenience let us split $B$ into the sets $B_i= \{(i,*,*,*) \}$ (here and in the sequel $*$ will be sometimes used as a wildcard in the tuples denoting elements of the domain). We will sometimes call $B_i$ the $i$-th \emph{layer} of $B$. 

\subsubsection{Some simple combinatorics: Extension functions} During the construction of the satisfaction forest we will design a special strategy for assigning 
labels to the leaves. To this end we introduce  an auxiliary combinatorial tool, which we will call
\emph{extension functions}.

Let us recall a well known Hall's marriage theorem.
A \emph{matching} in a bipartite graph $(G_1, G_2, E$) is a partial injective function $f:G_1 \rightarrow G_2$ such that if $f(a)=b$ then
$(a,b) \in E$. 
\begin{theorem}[Hall]
Let $(G_1, G_2, E)$ be a bipartite graph. There exists a matching covering $G_1$ iff for any set $W \subseteq G_1$ the number of vertices of
$G_2$ incident to the edges emitted from $W$ is greater or equal to $|W|$.
\end{theorem}

For a natural number $n$, let $[n]$ denote the set $\{1, \ldots, n \}$ and for $1 \le l \le n$ let $[n]^{l}$ denote the set of all subsets of $[n]$ of cardinality $l$.

\begin{lemma}
For every $0 <  l  <  K$ there exists a $1{-}1$ function $ext_l: [2K]^{l} \rightarrow [2K]^{l+1}$ such that for any $S \in [2K]^{l}$ we have that
	$S \subseteq ext_l(S)$.
\end{lemma}
\begin{proof}
Consider the bipartite graph $ ([2K]^{l}, [2K]^{l+1}, E)$ such that $(S, S') \in E$ iff $S \subseteq S'$. To show that a desired  $ext_l$ exists it
suffices to show the existence of a matching covering entirely the set $[2K]^{l}$. To this end we apply Hall's marriage theorem.
In our graph every 
node from $[2K]^{l}$ has degree $2K-l$ (given an $l$-element subset of $[2K]$ it can be expanded to an $l+1$ subset just by adding to it 
precisely one of the remaining $2K - l $ elements) and every node from $[2K]^{l+1}$ has degree $l+1$ (to obtain an $l$-element subset of a $l+1$-subset
one just removes one of the elements of the latter). Take a subset $W$ of $[2K]^{l}$.
The nodes of this subset are incident to $|W| \cdot (2K-l)$ edges in total. Let us see that the number of nodes in $[2K]^{l+1}$ incident to a node from $W$ is greater than or equal
to $| W |$. Indeed, assume to the contrary that it is not. Then at most $| W | - 1$ nodes absorb $|W| \cdot (2K-l)$ edges emitted by $W$,
but this means that  $|W| \cdot (2K-l) \le (|W|-1) \cdot (l+1)$. Rearranging this inequality we get that 
$|W|(2K-2l-1)  +l +1 \le 0$. But using the assumption that $0 < l < K$ we have that $(2K-2l-1)>0$ and hence
the whole left-hand side of the last inequality must be greater than $0$. Contradiction. Thus
our graph satisfies the Hall's theorem assumptions which guarantee the existence of a matching from $[2K]^{l}$ to $[2K]^{l+1}$, covering entirely $[2K]^{l}$.
This matching can be taken as $ext_l$. 
\end{proof}

Choose an extension function $ext_l$ for every $l$ and let $ext = \bigcup_{l=1}^{K-1} ext_l$, that is, $ext$ is a function which takes a non-empty subset of $[2K]$ of size at most $K-1$ and returns a superset 
containing precisely one new element. Obviously $ext$ remains an injective function.

\subsubsection{Construction of a satisfaction forest} We now describe how to construct a satisfaction forest $\cT_1, \ldots, \cT_{\mse}$ for $\phi$ over the domain $B$.
 It should be helpful to announce how we are going to take care of the  consistency conditions for the whole forest:
\begin{itemize} \itemsep0pt
\item Conditions (F2) and (T3): 
With every element  $a=(*,*,*,l) \in B$  we associate the $1$-type  $\alpha_l$. Whenever $a$ will be used as a label of a node in a satisfaction tree then
its $1$-type in the pre-structure defined for any branch containing a node labelled with $a$ will be set to $\alpha_l$.
\item Conditions (F3) and (T4): for a a pair of distinct branches $\flat_1$, $\flat_2$ (either belonging to the same tree or to two different trees) we will simply have $Set(\flat_1) \not= Set(\flat_2)$. This condition will be ensured by an appropriate use of the extension function. Here it is important that the last quantifier in every $\phie_i$-conjunct is existential, and hence the last node of every branch in $\mathcal{T}_i$  is also existential, so we can freely choose its label from $B$.

\end{itemize}

Let us explain how to construct a single $\cT_i$, a $\phie_i$-satisfaction tree over $B$. 
The general shape of $\cT_i$ is determined by $\phie_i$ and $B$: we know how many nodes we need, we know which of them
are existential, and which are universal, we know the labels of the universal nodes. It remains to assign labels to existential nodes (elements of $B$) and to branches (pre-structures on the set of elements formed from the labels of the nodes on a branch).

We define an auxiliary function $pat$ which for every node of $\cT_i$ returns \emph{a pattern element} from $\str{A}$. We will choose $pat(n)$, so that
its $1$-type is equal to type of $\cL(n)$. We remark, that if two nodes from different branches have the same label then they do not need to have the same pattern element.

Consider a node $n_k$ and assume that all its non-root ancestors $n_1, \dots, n_{k-1}$ have the function $pat$ and their labels already defined.
We proceed as follows

\begin{itemize} \itemsep0pt
\item If $n_k$ is universal then its label $\cL(n_k)$ is known
\begin{itemize} \itemsep0pt
	\item 
If $\cL(n_k)=\cL(n_j)$ for some $j<k$ then we set $pat(n_k)=pat(n_j)$. 
\item If the label $\cL(n_k)$ is not used by the ancestors of $n$ then choose as $pat(n)$ an arbitrary element of $\str{A}$ of the $1$-type assigned to
$\cL(n_k)$.
(In particular, we may use an element which was used by one of the ancestors of $n$)
\end{itemize}
\item If $n_k$ is existential then we need to define both $\cL(n)$ and $pat(n)$. By our construction we have  that 
\begin{eqnarray*}
\str{A} \models \exists x_k \Qfr_{k+1}^i x_{k+1} \ldots \Qfr_{k_i-1}^i x_{k_i-1}\exists x_{k_i} \psie_i (pat(n_1), \ldots, pat(n_{k-1}),x_k, x_{k+1}, \ldots, x_{k_i}).
\end{eqnarray*}

We choose an element $w \in A$ witnessing the previous formula, i.e., an element such that 
$$\str{A} \models \Qfr_{k+1}^i x_{k+1} \ldots \Qfr_{k_i-1}^i x_{k_i-1}\exists x_{k_i} \psie_i (pat(n_1), \ldots, pat(n_{k-1}), w, x_{k+1}, \ldots, x_{k_{i}})$$
 and set $pat(n_k)=w$. To define the label of $n_k$ we consider two cases:

\begin{itemize} \itemsep0pt
\item
If $n_k$ is not a leaf then:
\begin{itemize} \itemsep0pt
\item if $pat(n_j)=w$ for some $j<k$ then set $\cL(n_k)=\cL(n_j)$
\item otherwise we choose as $\cL(n_k)$ an arbitrary element of $B$ which has assigned the $1$-type $\tp^{\str{A}}(w)$, not used by the ancestors of $n_k$ (there are many copies
of each $1$-type in $B$ so it is always possible). 
\end{itemize}

\item
If $n_k$ is a leaf then let $\flat$ be the branch of $n_k$ and let  $S=\{j: n_l \in B_j \mbox{ for some } l<k \}$. Of course, $|S| <k \le K$ so $ext(S)$ is defined. Let $s$ be the unique member of $ext(S) \setminus S$.
We take as $\cL(n_k)$ an element $(s, i, t, l) \in B_s$ where $l$ is such that $\alpha_l=\tp^{\str{A}}(w)$, and where $t$ is chosen so that none of the  branches $\flat'$ 
of the current tree for
which the labels have been already defined 
such that $Set^-(\flat')=Set^-(\flat)$ used $(s, i, t, l)$ as the label of its leaf. 
We indeed have  enough elements for this, since obviously $|Set^-(\flat)| \le K-1$ and
thus there are at most $(K-1)^{K-1}$ different branches whose nodes from the first $K-1$ levels are
labelled by elements of $|Set^-(\flat)|$ (recall that there are $(K-1)^{K-1}$ possible choices for $t$).

\end{itemize}

\end{itemize}

Take now any branch $\flat$ of $\cT_i$. It remains to define the pre-structure $\cL(\flat)$. 
For any relational symbol $R$ of arity $m$ and any sequence $a_{i_1}, \ldots, a_{i_m}$ of elements of $Set(\flat)$ containing
all the elements of $Set(\flat)$ we set $R(a_{i_1}, \ldots, a_{i_m})$ to be true iff
$R(pat(a_{i_1}), \ldots, pat(a_{i_m}))$ is true in $\str{A}$. For every $a_j$ its $1$-type is set to be equal to the $1$-type of $pat(a_j)$.
 This completes the definition of the pre-structure on $Set(\flat)$. 
Note that this ensures that this pre-structure satisfies $\psi_i^{\exists}(Seq(\flat))$.

\subsection{Correctness}
Let us now see that the defined satisfaction forest indeed satisfies all the required conditions. 

\begin{itemize} \itemsep0pt
\item Conditions (T1), (T2) and (T3) should be clear.
\item For (T4) we show that  there is no pair of branches $\flat$, $\flat'$ in a tree $\cT_i$ with $Set(\flat)=Set(\flat')$. 
Indeed, we have chosen as labels of the leaves of $\flat$ and $\flat'$ two different elements $b=(s, i, x, *)$ and $b'=(s,i,y,*)$ of a layer $B_s$ which
is not inhabited by the elements of $Set^-(\flat)$ or $Set^-(\flat')$ (due to the use of the function $ext$). So $b \in Set(\flat)$ but $b \not\in Set(\flat')$
and thus $Set(\flat) \not=Set(\flat')$.
\item To show that (T5) holds assume to the contrary that for some branch $\flat$, $\str{H}^*=\cL(\flat)$ is not $\phiu$-compatible; take $i$ for
which it is not $\phiu_i$-compatible. So, for some sequence $a_1, \ldots, a_{l_1}$ such that $\{a_1, \ldots, a_{l_1}\}=Set(\flat)=H$ we have
$\str{H}^* \not\models \psiu_i ( a_1, \ldots, a_{l_i})$. But then the definition of the pre-structure in $\cL(\flat)$ implies that $\str{A} \not\models \psiu_i(pat(a_1), \ldots, pat(a_{l_i}))$, that is $\str{A}$ violates $\phiu_i$. Contradiction.
   
\item Conditions (F1), (F2) should be clear.
\item For (F3) the argument is similar to the argument for (T4): We show that  there is is pair of branches $\flat_1$, $\flat_2$ in a tree $\cT_i$, and resp., $\cT_j$, $i\not=j$, with $Set(\flat_1)=Set(\flat_2)$. Again, this follows from the fact that we have chosen as labels of the leaves of $\flat_1$ and $\flat_2$ two different elements $b$ and $b'$ of a layer $B_s$ which
is not inhabited by the elements of $Set^-(\flat_1)$ or $Set^-(\flat_2)$. This time the elements $b_1$ and $b_2$ are different from each other since
$b_1=(s, i, *, *)$ and $b_2=(s, j, *,*)$. 
\item For (F4) we reason precisely as in the reasoning for (F4) in the proof of  the Claim in Section \ref{s:satfor} (we just replace the structure $\str{B}$
from this proof with the currently considered structure $\str{A}$). 
\end{itemize}

An immediate consequence of Thm.~\ref{t:fmp} is: 
\begin{theorem}
The satisfiability problem for \UFAm{} is \NExpTime-complete.
\end{theorem}
\begin{proof}
The lower bound is inherited from the lower bound for \FOt{} \cite{Lew80}.
Let us turn to the upper bound.

By Lemma \ref{l:nf2} it suffices to show how to decide satisfiability of a normal form sentence $\varphi$. By Theorem \ref{t:fmp} if $\varphi$ is satisfiable then it has a model with exponentially
bounded domain. We guess some  natural description of such a model $\str{A}$. We note
that this description is also of exponential size with respect to $|\varphi|$: Indeed, we need to describe some number (linearly bounded in $|\varphi|$)
of relations of arity at most $|\varphi|$, and it is straightforward, taking into consideration the size of the domain, that a description of a single such relation is at most exponential in $|\varphi|$.
A verification of a single normal form conjunct in the guessed structure can be done in an exhaustive way, by considering all possible substitutions
for the variables.

Alternatively, instead of guessing a model one could guess a satisfaction forest for $\varphi$. Again, a routine inspection reveals that the size of its
description can be bounded exponentially in $|\varphi|$; also the verification of the properties (T1)-(T5), (F1)-(F4) would not be problematic.
\end{proof}

\section{Infinity axiom with free use of equality}

In this section we note that allowing for free use of equality in our logic changes the situation significantly: we lose the finite model property. We recall that in the case of \UF{} free use of equality does not spoil the decidability and even does not change the complexity.

In the recent paper \cite{FK23} we note that the fragment with arbitrary blocks of quantifiers  \UFA{}  and with free use of equality contains infinity axioms (satisfiable formulas without finite models), by constructing the following three-variable formula:
$$\exists x S(x) \wedge \forall x \exists y \forall z ( \neg S(y) \wedge R(x,y,z)  \wedge (x=z \vee \neg R(z,y,x))),$$
which has no finite models but is satisfied in the model whose universe is the set of natural numbers, $S$ is true only at $0$ and $Rxyz$ is 
true iff $y=x+1$.

The above example can be simply adapted to the case of \UFAm{} with free use of equality. We just add a dummy existentially quantified variable $t$ and
require it to be equal to the previous, universally quantified variable $z$. To accommodate all the variables we increase the arity of $R$ by $2$ (one can think
that the first and the last position of $R$ from the previous example have been doubled):

$$\exists x S(x) \wedge \forall x \exists y \forall z \exists t. (t=z \wedge \neg S(y) \wedge R(x,x,y,z,t)  \wedge (x=z \vee \neg R(z,t,y,x,x))).$$

\section{Conclusions}
We identified a non-trivial uniform one-dimensional logic in which a use of mixed blocks of quantifiers is allowed,
strictly extending the two-variable fragment \FOt{} without equality and the previously defined fragment \sUF{} without equality.
We proved that, similarly to  \FOt{} and s\UF{},  this logic has the finite, exponential model property and \NExpTime-complete satisfiability problem.

There are two interesting directions, orthogonal to each other, in which it would be 
valuable to extend our work. The first is investigating the decidability, complexity
and the status of the finite model property for full \UFA{} without equality,
that is to see what happens to our logic if arbitrary blocks of quantifiers, possibly ending with the universal quantifier,
are allowed. As already mentioned, in our recent work \cite{FK23} we answered this question for the three variable restriction, \UFAthree{}, of \UFA{} by
showing the exponential model property and \NExpTime-completeness of its satisfiability problem.

The second idea is to revive the research on Maslov Class $\bar{\text{K}}$, by attempting to determine the precise complexity of its satisfiability problem
and investigating whether it has the finite model property. When designing the fragment \UFAm{} we took some inspiration from
the definition of $\bar{\text{K}}$, and indeed we were able to  reduce satisfiability of the former to the latter. We believe that 
what we have learned working on \UFAm{} will prove useful in the case of $\bar{\text{K}}$.

There are also some probably slightly less attractive, but still interesting, a bit more technical questions that one can try to answer. For example,
what happens to our logic if a use of equalities/inequalities (free or uniform) or constants is allowed.

\section*{Acknowledgement} This work is supported by  NCN grant No. 2021/41/B/ ST6/00996.

\bibliographystyle{eptcs}
\bibliography{bibshort}

\end{document}